\documentclass[letterpaper, 10 pt, conference]{ieeeconf}  

\usepackage{amsmath}
\usepackage{graphicx}
\usepackage{cite}
\usepackage{balance}
\usepackage{amssymb,amsfonts}

\usepackage{amssymb,amsfonts}

\usepackage{lipsum}
\usepackage{mathtools}
\usepackage{cuted}

\usepackage{balance}

\newtheorem{lemma}{\it Lemma}
\newtheorem{definition}{\it Definition}

\newtheorem{proposition}{\it Proposition}

\IEEEoverridecommandlockouts                              

\title{\LARGE \bf
	Independent Gaussian Distributions Minimize the Kullback--Leibler (KL) Divergence from Independent Gaussian Distributions
}

\author{Song Fang$^{1}$ and Quanyan Zhu$^{1}$
\thanks{$^{1}$ Song Fang and Quanyan Zhu are with the Department of Electrical and Computer Engineering, New York University, New York, USA
        {\tt\small song.fang@nyu.edu; quanyan.zhu@nyu.edu}}%
}

\begin{document}

\maketitle
\thispagestyle{empty}
\pagestyle{empty}

\begin{abstract}
	
	This short note is on a property of the Kullback--Leibler (KL) divergence which indicates that independent Gaussian distributions minimize the KL divergence from given independent Gaussian distributions. The primary purpose of this note is for the referencing of papers that need to make use of this property entirely or partially.
\end{abstract}


\section{Introduction}

The Kullback--Leibler (KL) divergence was proposed in \cite{kullback1951information} (see also \cite{kullback1997information}), and has been employed in various research areas, including, e.g., information theory \cite{Cov:06}, signal processing \cite{Kay20}, statistics \cite{pardo2006statistical}, control and estimation theory \cite{lindquist2015linear}, and machine learning \cite{goodfellow2016deep}.

In this brief note, we present the property of the KL divergence that
independent Gaussian distributions minimize the KL divergence from given independent Gaussian distributions, while characterizing the corresponding minimum KL divergence in terms of covariance matrices. 
The primary purpose of this note is for the referencing of future papers that need to utilize this property entirely (see Proposition~\ref{c1}) or partially (see Proposition~\ref{p1}, Proposition~\ref{p2}, and Proposition~\ref{p3}), such as \cite{KLarxiv}. 

%
%
%
%
%

\section{Preliminaries}

Throughout the paper, we consider zero-mean real-valued continuous random variables and random vectors, as well as discrete-time stochastic processes. We represent random variables and random vectors using boldface letters, e.g., $\mathbf{x}$, while the probability density function of $\mathbf{x}$ is denoted as $p_\mathbf{x}$. We denote  random vector  $\left[ \mathbf{x}_0^{\mathrm{T}},\ldots,\mathbf{x}_{k}^{\mathrm{T}} \right]^{\mathrm{T}}$ by $\mathbf{x}_{0,\ldots,k}$ for simplicity. 
Below we introduce the definitions of entropy and conditional entropy (see, e.g., \cite{Cov:06}, where properties of these notions can also be found).

\begin{definition} The differential entropy of random vector $\mathbf{x}$ with density $p_{\mathbf{x}}$ is defined as
	\begin{flalign}
	h\left( \mathbf{x} \right)
	=-\int p_{\mathbf{x}} \left(\mathbf{u}\right) \ln p_{\mathbf{x}} \left(\mathbf{u}\right) \mathrm{d} \mathbf{u}. \nonumber
	\end{flalign}
	The conditional differential entropy of random vector $\mathbf{x}$ given random vector $\mathbf{y}$ with joint density $p_{\mathbf{x}, \mathbf{y}} $ and conditional density $p_{\mathbf{x} | \mathbf{y}} $ is defined as
	\begin{flalign}
	h\left(\mathbf{x}\middle|\mathbf{y}\right)
	=-\int p_{\mathbf{x}, \mathbf{y}} \left(\mathbf{u},\mathbf{v}\right)\ln p_{\mathbf{x} | \mathbf{y}} \left(\mathbf{u},\mathbf{v}\right) \mathrm{d}\mathbf{u}\mathrm{d}\mathbf{v}. \nonumber
	\end{flalign}
\end{definition}

\vspace{3mm}

The KL divergence (see, e.g., \cite{kullback1951information}) is defined as follows.

\begin{definition}
	Consider $m$-dimensional random vectors $\mathbf{x}$ and $\mathbf{y}$ with densities $p_\mathbf{x}$ and $p_\mathbf{y}$, respectively.
	The KL divergence from distribution $p_\mathbf{x}$
	to distribution $p_\mathbf{y}$ is defined as
	\begin{flalign}
	\mathrm{KL} \left(p_{\mathbf{y}} \| p_{\mathbf{x}} \right)
	= \int p_{\mathbf{y}} \left( \mathbf{u} \right) \ln \frac{p_{\mathbf{y}} \left( \mathbf{u} \right)}{ p_{ \mathbf{x} } \left( \mathbf{u} \right)} \mathrm{d} \mathbf{u}. \nonumber
	\end{flalign}
\end{definition}

\vspace{3mm}

The next lemma (see, e.g., \cite{Kay20}) provides an explicit expression of KL divergence in terms of covariance matrices for Gaussian random vectors; note that herein and in the sequel, all random variables and random vectors are assumed to be zero-mean.

\begin{lemma} \label{Gaussian}
	Consider $m$-dimensional Gaussian random vectors $\mathbf{x}$ and $\mathbf{y}$ with covariance matrices $\Sigma_\mathbf{x}$ and $\Sigma_\mathbf{y}$, respectively.
	The KL divergence from distribution $p_\mathbf{x}$
	to distribution $p_\mathbf{y}$ is given by
	\begin{flalign}
	\mathrm{KL} \left( p_{\mathbf{y}} \| p_{\mathbf{x}} \right) = \frac{1}{2} \left[ \mathrm{tr} \left( \Sigma_{\mathbf{y}} \Sigma_{\mathbf{x}}^{-1} \right) - \ln \det \left( \Sigma_{\mathbf{y}} \Sigma_{\mathbf{x}}^{-1} \right) - m \right].
	\nonumber
	\end{flalign}
\end{lemma}

\vspace{3mm}

In the scalar case (when $m=1$), Lemma~\ref{Gaussian} reduces to the following formula for Gaussian random variables:
\begin{flalign}
\mathrm{KL} \left( p_{\mathbf{y} } \| p_{\mathbf{x} } \right) = \frac{1}{2} \left[  \frac{ \sigma_{\mathbf{y} }^2}{\sigma_{\mathbf{x} }^2} 
-  \ln  \left( \frac{ \sigma_{\mathbf{y} }^2 }{\sigma_{\mathbf{x} }^2} \right) - 1 \right].
\nonumber
\end{flalign}

\section{KL Divergence-Minimizing Distributions}

Let us state the main result of this note first.

\begin{proposition} \label{c1}
	Consider $m$-dimensional random vectors $\mathbf{x}$ and $\mathbf{y}$ with positive definite covariance matrices $\Sigma_{\mathbf{x}}$ and $\Sigma_{\mathbf{y}}$, respectively. Suppose that $\mathbf{x}$ is Gaussian, whereas $\mathbf{y}$ is not necessarily Gaussian. In addition, suppose that $\Sigma_{\mathbf{x}}$ is a diagonal matrix, i.e.,
	\begin{flalign}
	\Sigma_{\mathbf{x}}
	= \Lambda_{\mathbf{x}}
	= \mathrm{diag} \left( \sigma_{\mathbf{x} \left( 1 \right)}^2, \ldots, \sigma_{\mathbf{x} \left( m \right)}^2 \right).
	\end{flalign}
	Denote the diagonal terms of $\Sigma_{\mathbf{y}}$ by $\sigma_{\mathbf{y} \left( 1 \right)}^2, \ldots, \sigma_{\mathbf{y} \left( m \right)}^2$, whereas $\Sigma_{\mathbf{y}}$ is not necessarily a diagonal matrix.
	Then,
	\begin{flalign}
	\mathrm{KL} \left( p_{\mathbf{y}} \| p_{\mathbf{x}} \right)
	\geq \frac{1}{2} \left\{ \sum_{i=1}^{m} \left[ \frac{ \sigma_{\mathbf{y} \left( i \right)}^2}{\sigma_{\mathbf{x} \left( i \right)}^2} \right]
	-  \sum_{i=1}^m \ln  \left[ \frac{ \sigma_{\mathbf{y} \left( i \right)}^2 }{\sigma_{\mathbf{x} \left( i \right)}^2} \right] - m \right\},
	\end{flalign}
	where equality holds if $\mathbf{y}$ is Gaussian while $\Sigma_{\mathbf{y}}$ is a diagonal matrix as
	\begin{flalign}
	\Sigma_{\mathbf{y}}
	= \Lambda_{\mathbf{y}}
	= \mathrm{diag} \left( \sigma_{\mathbf{y} \left( 1 \right)}^2, \ldots, \sigma_{\mathbf{y} \left( m \right)}^2 \right).
	\end{flalign}
\end{proposition}

\vspace{3mm}

Before we prove Proposition~\ref{c1}, we shall present the following three properties of the KL divergence in Proposition~\ref{p1}, Proposition~\ref{p2}, and Proposition~\ref{p3}, respectively, based upon which Proposition~\ref{c1} can then be established.

The first property indicates that Gaussian distribution minimizes, among all distributions with the same covariance, the KL divergence from a given Gaussian distribution. Note that although this property has been proved in the literature, e.g., recently in \cite{zhang2017stealthy} (see Lemma~1 therein), we still provide a proof below for the self-containment and consistence in notation of this paper. 

\begin{proposition} \label{p1}
	Consider $m$-dimensional random vectors $\mathbf{x}, \mathbf{y}$. Suppose that $\mathbf{x}$ is Gaussian.
	Then,
	\begin{flalign}
	\mathrm{KL} \left( p_{\mathbf{y}} \| p_{\mathbf{x}} \right)
	\geq \mathrm{KL} \left( p_{\mathbf{y}^{\mathrm{G}}} \| p_{\mathbf{x}} \right),
	\nonumber
	\end{flalign}
	where $\mathbf{y}^{\mathrm{G}}$ denotes a Gaussian random vector with the same covariance as $\mathbf{y}$, and equality holds if $\mathbf{y}$ is Gaussian.
\end{proposition}

\begin{proof} Note that
	\begin{flalign}
	&\mathrm{KL} \left(p_{\mathbf{y}} \| p_{\mathbf{x}} \right)
	= \int p_{\mathbf{y}} \left( \mathbf{u} \right) \ln \frac{p_{\mathbf{y}} \left( \mathbf{u} \right)}{ p_{ \mathbf{x} } \left( \mathbf{u} \right)} \mathrm{d} \mathbf{u} \nonumber \\
	& = \int p_{\mathbf{y}} \left( \mathbf{u} \right) \ln p_{\mathbf{y}} \left( \mathbf{u} \right) \mathrm{d} \mathbf{u}
	- \int p_{\mathbf{y}} \left( \mathbf{u} \right) \ln p_{\mathbf{x}} \left( \mathbf{u} \right) \mathrm{d} \mathbf{u} \nonumber \\
	& = - h \left(\mathbf{y} \right)
	- \int p_{\mathbf{y}} \left( \mathbf{u} \right) \ln \left[ \frac{1}{ \left( 2 \pi \right)^{\frac{m}{2}} \left( \det  \Sigma_{\mathbf{x}} \right)^{\frac{1}{2}} } \mathrm{e}^{ - \frac{1}{2} \mathbf{u}^{\mathrm{T}} \Sigma_{\mathbf{x}}^{-1} \mathbf{u} } \right] \mathrm{d} \mathbf{u} \nonumber \\
	& = - h \left(\mathbf{y} \right)
	- \int p_{\mathbf{y}} \left( \mathbf{u} \right) \ln \left( \mathrm{e}^{ - \frac{1}{2} \mathbf{u}^{\mathrm{T}} \Sigma_{\mathbf{x}}^{-1} \mathbf{u} } \right) \mathrm{d} \mathbf{u} \nonumber \\
	&~~~~ - \ln \left[ \frac{1}{ \left( 2 \pi \right)^{\frac{m}{2}} \left( \det  \Sigma_{\mathbf{x}} \right)^{\frac{1}{2}} }  \right] \nonumber \\
	\nonumber \\
	& = - h \left(\mathbf{y} \right)
	+ \frac{1}{2} \int p_{\mathbf{y}} \left( \mathbf{u} \right)  \mathbf{u}^{\mathrm{T}} \Sigma_{\mathbf{x}}^{-1} \mathbf{u}  \mathrm{d} \mathbf{u} + \ln \left[  \left( 2 \pi \right)^{\frac{m}{2}} \left( \det  \Sigma_{\mathbf{x}} \right)^{\frac{1}{2}}   \right] \nonumber \\
	& = - h \left(\mathbf{y} \right)
	+ \frac{1}{2} \mathbb{E} \left( \mathbf{y}^{\mathrm{T}} \Sigma_{\mathbf{x}}^{-1} \mathbf{y} \right) + \ln \left[  \left( 2 \pi \right)^{\frac{m}{2}} \left( \det  \Sigma_{\mathbf{x}} \right)^{\frac{1}{2}}   \right] \nonumber \\
	&= - h \left(\mathbf{y} \right)
	+ \frac{1}{2} \mathrm{tr} \left[ \mathbb{E} \left( \mathbf{y} \mathbf{y}^{\mathrm{T}} \Sigma_{\mathbf{x}}^{-1} \right) \right] + \ln \left[  \left( 2 \pi \right)^{\frac{m}{2}} \left( \det  \Sigma_{\mathbf{x}} \right)^{\frac{1}{2}}   \right] \nonumber \\
	&= - h \left(\mathbf{y} \right)
	+ \frac{1}{2} \mathrm{tr} \left[ \mathbb{E} \left( \mathbf{y} \mathbf{y}^{\mathrm{T}}  \right) \Sigma_{\mathbf{x}}^{-1} \right] + \ln \left[  \left( 2 \pi \right)^{\frac{m}{2}} \left( \det  \Sigma_{\mathbf{x}} \right)^{\frac{1}{2}}   \right] \nonumber \\
	& = - h \left(\mathbf{y} \right)
	+ \frac{1}{2} \mathrm{tr} \left( \Sigma_{\mathbf{y}} \Sigma_{\mathbf{x}}^{-1} \right) + \ln \left[  \left( 2 \pi \right)^{\frac{m}{2}} \left( \det  \Sigma_{\mathbf{x}} \right)^{\frac{1}{2}}   \right] \nonumber \\
	& \geq - h \left(\mathbf{y}^{\mathrm{G}} \right)
	+ \frac{1}{2} \mathrm{tr} \left( \Sigma_{\mathbf{y}} \Sigma_{\mathbf{x}}^{-1} \right) + \ln \left[  \left( 2 \pi \right)^{\frac{m}{2}} \left( \det  \Sigma_{\mathbf{x}} \right)^{\frac{1}{2}}   \right], \nonumber
	\end{flalign}
	where $\mathbf{y}^{\mathrm{G}}$ denotes a Gaussian random vector with the same covariance as $\mathbf{y}$, and equality holds if $\mathbf{y}$ is Gaussian. Meanwhile, similarly (in a backward manner) to the previous derivations, it can be verified that 
	\begin{flalign}
	&- h \left(\mathbf{y}^{\mathrm{G}} \right)
	+ \frac{1}{2} \mathrm{tr} \left( \Sigma_{\mathbf{y}} \Sigma_{\mathbf{x}}^{-1} \right) + \ln \left[  \left( 2 \pi \right)^{\frac{m}{2}} \left( \det  \Sigma_{\mathbf{x}} \right)^{\frac{1}{2}}   \right] \nonumber \\
	&~~~~ = \mathrm{KL} \left( p_{\mathbf{y}^{\mathrm{G}}} \| p_{\mathbf{x}} \right), \nonumber
	\end{flalign}
	which completes the proof.
\end{proof}

The second property mandates that independent distributions minimize the KL divergence from given independent distributions. In comparison, the mutual information (or entropy) counterpart of this property can be found in \cite{fang2017towards} (see Lemma~2.5 therein).

\begin{proposition} \label{p2}
	Consider $m_1$-dimensional  random vectors $\mathbf{x}_{1}, \mathbf{y}_{1}$ and $m_2$-dimensional $\mathbf{x}_{2}, \mathbf{y}_{2}$. Suppose that $\mathbf{x}_{1}$ and $\mathbf{x}_{2}$ are independent.
	Then,
	\begin{flalign}
	\mathrm{KL} \left( p_{\mathbf{y}_{1}, \mathbf{y}_{2}} \| p_{\mathbf{x}_{1}, \mathbf{x}_{2}} \right)
	\geq \mathrm{KL} \left( p_{\mathbf{y}_{1}} \| p_{\mathbf{x}_{1}} \right) + \mathrm{KL} \left( p_{\mathbf{y}_{2}} \| p_{\mathbf{x}_{2}} \right),
	\end{flalign}
	where equality holds if $\mathbf{y}_{1}$ and $\mathbf{y}_{2}$ are independent.
	
	More generally, consider $m_1$-dimensional random vectors $\mathbf{x}_{1}, \mathbf{y}_{1}$, $m_2$-dimensional $\mathbf{x}_{2}, \mathbf{y}_{2}$, up to $m_n$-dimensional $\mathbf{x}_{n}, \mathbf{y}_{n}$. Suppose that $\mathbf{x}_{1}, \mathbf{x}_{2}, \ldots, \mathbf{x}_{n}$ are mutually independent.
	Then,
	\begin{flalign}
	\mathrm{KL} \left( p_{\mathbf{y}_{1, 2, \ldots, n}} \| p_{\mathbf{x}_{1, 2, \ldots, n}} \right)
	\geq \sum_{i=1}^{n} \mathrm{KL} \left( p_{\mathbf{y}_{i}}\| p_{\mathbf{x}_{i}} \right),
	\nonumber
	\end{flalign}
	where equality holds if $\mathbf{y}_{1}, \mathbf{y}_{2}, \ldots, \mathbf{y}_{n}$ are mutually independent.
\end{proposition}

\begin{proof} Note first that
	\begin{flalign}
	&\mathrm{KL} \left( p_{\mathbf{y}_{1}, \mathbf{y}_{2}} \| p_{\mathbf{x}_{1}, \mathbf{x}_{2}} \right) \nonumber \\
	&~~~~ =  \iint p_{\mathbf{y}_{1}, \mathbf{y}_{2}} \left( \mathbf{u}, \mathbf{v} \right) \ln \frac{ p_{ \mathbf{y}_{1}, \mathbf{y}_{2} } \left( \mathbf{u}, \mathbf{v} \right)}{p_{\mathbf{x}_{1}, \mathbf{x}_{2}} \left( \mathbf{u}, \mathbf{v} \right)} \mathrm{d} \mathbf{u}
	\mathrm{d} \mathbf{v}
	\nonumber \\
	&~~~~ =  \iint p_{\mathbf{y}_{1}, \mathbf{y}_{2}} \left( \mathbf{u}, \mathbf{v} \right) \ln \frac{
		p_{ \mathbf{y}_{1} } \left( \mathbf{u} \right)
		p_{ \mathbf{y}_{2} | \mathbf{y}_{1} } \left( \mathbf{u}, \mathbf{v} \right)}{p_{\mathbf{x}_{1}} \left( \mathbf{u} \right) p_{\mathbf{x}_{2}| \mathbf{x}_{1}} \left( \mathbf{u}, \mathbf{v} \right)} \mathrm{d} \mathbf{u}
	\mathrm{d} \mathbf{v} 
	\nonumber \\
	&~~~~ =  \iint p_{\mathbf{y}_{1}, \mathbf{y}_{2}} \left( \mathbf{u}, \mathbf{v} \right) \ln \frac{
		p_{ \mathbf{y}_{1} } \left( \mathbf{u} \right)}{p_{\mathbf{x}_{1}} \left( \mathbf{u} \right)} \mathrm{d} \mathbf{u}
	\mathrm{d} \mathbf{v} \nonumber \\
	&~~~~~~~~ +  \iint p_{\mathbf{y}_{1}, \mathbf{y}_{2}} \left( \mathbf{u}, \mathbf{v} \right) \ln \frac{
		p_{ \mathbf{y}_{2} | \mathbf{y}_{1} } \left( \mathbf{u}, \mathbf{v} \right)}{ p_{\mathbf{x}_{2}| \mathbf{x}_{1}} \left( \mathbf{u}, \mathbf{v} \right)} \mathrm{d} \mathbf{u}
	\mathrm{d} \mathbf{v}
	\nonumber \\
	&~~~~ =  \int p_{\mathbf{y}_{1}} \left( \mathbf{u} \right) \ln \frac{
		p_{ \mathbf{y}_{1} } \left( \mathbf{u} \right)}{p_{\mathbf{x}_{1}} \left( \mathbf{u} \right)} \mathrm{d} \mathbf{u} \nonumber \\
	&~~~~~~~~ +  \iint p_{\mathbf{y}_{1}, \mathbf{y}_{2}} \left( \mathbf{u}, \mathbf{v} \right) \ln \frac{
		p_{ \mathbf{y}_{2} | \mathbf{y}_{1} } \left( \mathbf{u}, \mathbf{v} \right)}{ p_{\mathbf{x}_{2}| \mathbf{x}_{1}} \left( \mathbf{u}, \mathbf{v} \right)} \mathrm{d} \mathbf{u}
	\mathrm{d} \mathbf{v}
	\nonumber \\
	&~~~~ =  \mathrm{KL} \left( p_{\mathbf{y}_{1}} \| p_{\mathbf{x}_{1}} \right) \nonumber \\
	&~~~~~~~~ +  \iint p_{\mathbf{y}_{1}, \mathbf{y}_{2}} \left( \mathbf{u}, \mathbf{v} \right) \ln  p_{\mathbf{y}_{2}| \mathbf{y}_{1}} \left( \mathbf{u}, \mathbf{v} \right)
	\mathrm{d} \mathbf{u}
	\mathrm{d} \mathbf{v} \nonumber \\
	&~~~~~~~~ -  \iint p_{\mathbf{y}_{1}, \mathbf{y}_{2}} \left( \mathbf{u}, \mathbf{v} \right) \ln	p_{ \mathbf{x}_{2} | \mathbf{x}_{1} } \left( \mathbf{u}, \mathbf{v} \right) \mathrm{d} \mathbf{u}
	\mathrm{d} \mathbf{v}
	\nonumber \\
	&~~~~ =  \mathrm{KL} \left( p_{\mathbf{y}_{1}} \| p_{\mathbf{x}_{1}} \right) - h \left( \mathbf{y}_{2} | \mathbf{y}_{1} \right) \nonumber \\
	&~~~~~~~~ -  \iint p_{\mathbf{y}_{1}, \mathbf{y}_{2}} \left( \mathbf{u}, \mathbf{v} \right) \ln	p_{ \mathbf{x}_{2} | \mathbf{x}_{1} } \left( \mathbf{u}, \mathbf{v} \right) \mathrm{d} \mathbf{u}
	\mathrm{d} \mathbf{v}
	.
	\nonumber
	\end{flalign}
	Then, since $\mathbf{x}_{1}$ and $\mathbf{x}_{2}$ are independent, we have 
	\begin{flalign}
	p_{\mathbf{x}_{2}| \mathbf{x}_{1}} \left( \mathbf{u}, \mathbf{v} \right) = p_{\mathbf{x}_{2}} \left( \mathbf{v} \right), \nonumber
	\end{flalign}
	and thus
	\begin{flalign}
	&\mathrm{KL} \left( p_{\mathbf{y}_{1}} \| p_{\mathbf{x}_{1}} \right) - h \left( \mathbf{y}_{2} | \mathbf{y}_{1} \right) \nonumber \\
	&~~~~ -  \iint p_{\mathbf{y}_{1}, \mathbf{y}_{2}} \left( \mathbf{u}, \mathbf{v} \right) \ln	p_{ \mathbf{x}_{2} | \mathbf{x}_{1} } \left( \mathbf{u}, \mathbf{v} \right) \mathrm{d} \mathbf{u}
	\mathrm{d} \mathbf{v} \nonumber \\
	&~~~~ = \mathrm{KL} \left( p_{\mathbf{y}_{1}} \| p_{\mathbf{x}_{1}} \right) - h \left( \mathbf{y}_{2} | \mathbf{y}_{1} \right) \nonumber \\
	&~~~~ -  \iint p_{\mathbf{y}_{1}, \mathbf{y}_{2}} \left( \mathbf{u}, \mathbf{v} \right) \ln	p_{ \mathbf{x}_{2} } \left( \mathbf{v} \right) \mathrm{d} \mathbf{u}
	\mathrm{d} \mathbf{v}
	\nonumber \\
	&~~~~ = \mathrm{KL} \left( \mathbf{y}_{1} \| \mathbf{x}_{1} \right) - h \left( \mathbf{y}_{2} | \mathbf{y}_{1} \right) -  \int p_{\mathbf{y}_{2}} \left( \mathbf{v} \right) \ln	p_{ \mathbf{x}_{2} } \left( \mathbf{v} \right)
	\mathrm{d} \mathbf{v}
	\nonumber \\
	&~~~~ \geq \mathrm{KL} \left( \mathbf{y}_{1} \| \mathbf{x}_{1} \right) - h \left( \mathbf{y}_{2} \right) -  \int p_{\mathbf{y}_{2}} \left( \mathbf{v} \right) \ln	p_{ \mathbf{x}_{2} } \left( \mathbf{v} \right)
	\mathrm{d} \mathbf{v}
	, \nonumber
	\end{flalign}
	where equality holds if $\mathbf{y}_{1}$ and $\mathbf{y}_{2}$ are independent. 
	On the other hand,
	\begin{flalign}
	&- h \left( \mathbf{y}_{2} \right) -  \int p_{\mathbf{y}_{2}} \left( \mathbf{v} \right) \ln	p_{ \mathbf{x}_{2} } \left( \mathbf{v} \right)
	\mathrm{d} \mathbf{v}  \nonumber \\
	&~~~~ = \int p_{\mathbf{y}_{2}} \left( \mathbf{v} \right) \ln	p_{ \mathbf{y}_{2} } \left( \mathbf{v} \right)
	\mathrm{d} \mathbf{v} 
	-  \int p_{\mathbf{y}_{2}} \left( \mathbf{v} \right) \ln	p_{ \mathbf{x}_{2} } \left( \mathbf{v} \right)
	\mathrm{d} \mathbf{v} \nonumber \\
	&~~~~ = \int p_{\mathbf{y}_{2}} \left( \mathbf{v} \right) \ln	\frac{p_{ \mathbf{y}_{2} } \left( \mathbf{v} \right)}{p_{ \mathbf{x}_{2} } \left( \mathbf{v} \right)}
	\mathrm{d} \mathbf{v}
	= \mathrm{KL} \left( p_{\mathbf{y}_{2}} \| p_{\mathbf{x}_{2}} \right)
	. \nonumber
	\end{flalign}
	Consequently,
	\begin{flalign}
	\mathrm{KL} \left( p_{\mathbf{y}_{1}, \mathbf{y}_{2}} \| p_{\mathbf{x}_{1}, \mathbf{x}_{2}} \right)
	\geq \mathrm{KL} \left( p_{\mathbf{y}_{1}} \| p_{\mathbf{x}_{1}} \right) + \mathrm{KL} \left( p_{\mathbf{y}_{2}} \| p_{\mathbf{x}_{2}} \right), \nonumber
	\end{flalign}
	where equality holds if $\mathbf{y}_{1}$ and $\mathbf{y}_{2}$ are independent.
	
	More generally, since 
	$\mathbf{x}_{1}, \mathbf{x}_{2}, \ldots, \mathbf{x}_{n}$ are mutually independent, we have
	\begin{flalign}
	&\mathrm{KL} \left( p_{\mathbf{y}_{1,\ldots,n}} \| p_{\mathbf{x}_{1,\ldots,n}} \right) \nonumber \\
	&~~~~ \geq \mathrm{KL} \left( p_{\mathbf{y}_{1,\ldots,n-1}}\| p_{\mathbf{x}_{1,\ldots,n-1}} \right) + \mathrm{KL} \left( p_{\mathbf{y}_{n}}\| p_{\mathbf{x}_{n}} \right)\nonumber \\
	&~~~~ \geq \mathrm{KL} \left( p_{\mathbf{y}_{0,\ldots,n-2}}\| p_{\mathbf{x}_{0,\ldots,n-2}} \right) + \mathrm{KL} \left( p_{\mathbf{y}_{n-1}}\| p_{\mathbf{x}_{n-1}} \right) \nonumber \\
	&~~~~~~~~ + \mathrm{KL} \left( p_{\mathbf{y}_{n}}\| p_{\mathbf{x}_{n}} \right)\nonumber \\
	&~~~~ \geq \cdots 
	\geq \sum_{i=1}^{n} \mathrm{KL} \left( p_{\mathbf{y}_{i}}\| p_{\mathbf{x}_{i}} \right),
	\nonumber
	\end{flalign}
	and
	\begin{flalign}
	\mathrm{KL} \left( p_{\mathbf{y}_{1,\ldots,n}} \| p_{\mathbf{x}_{1,\ldots,n}} \right)
	= \sum_{i=1}^{n} \mathrm{KL} \left( p_{\mathbf{y}_{i}}\| p_{\mathbf{x}_{i}} \right)
	\nonumber
	\end{flalign}
	if $\mathbf{y}_{1}, \mathbf{y}_{2}, \ldots, \mathbf{y}_{n}$ are mutually independent.
\end{proof}

The third property is a direct consequence of the second property; in fact, it is a property of positive definite matrices, as will be discussed after Proposition~\ref{p3}. To compare with, the mutual information (or entropy) counterpart of this property can be found in \cite{fang2017towards} (see Lemma~2.9 therein).

\begin{proposition} \label{p3}
	Consider $m$-dimensional  Gaussian random vectors $\mathbf{x}$ and $\mathbf{y}$ with positive definite covariance matrices $\Sigma_{\mathbf{x}}$ and $\Sigma_{\mathbf{y}}$, respectively. Suppose that $\Sigma_{\mathbf{x}}$ is a diagonal matrix, i.e.,
	\begin{flalign}
	\Sigma_{\mathbf{x}}
	= \Lambda_{\mathbf{x}}
	= \mathrm{diag} \left( \sigma_{\mathbf{x} \left( 1 \right)}^2, \ldots, \sigma_{\mathbf{x} \left( m \right)}^2 \right).
	\end{flalign}
	Denote the diagonal terms of $\Sigma_{\mathbf{y}}$ by $\sigma_{\mathbf{y} \left( 1 \right)}^2, \ldots, \sigma_{\mathbf{y} \left( m \right)}^2$; note that $\Sigma_{\mathbf{y}}$ is not necessarily a diagonal matrix.
	Then,
	\begin{flalign}
	&\frac{1}{2} \left[ \mathrm{tr} \left( \Sigma_{\mathbf{y}} \Sigma_{\mathbf{x}}^{-1} \right) - \ln \det \left( \Sigma_{\mathbf{y}} \Sigma_{\mathbf{x}}^{-1} \right) - m \right] \nonumber \\
	&~~~~ = \mathrm{KL} \left( p_{\mathbf{y}} \| p_{\mathbf{x}} \right)
	\nonumber \\
	&~~~~ \geq \sum_{i=1}^{m} \mathrm{KL} \left( p_{\mathbf{y} \left( i \right)} \| p_{\mathbf{x} \left( i \right)} \right) \nonumber \\
	&~~~~ = \frac{1}{2} \left\{ \sum_{i=1}^{m} \left[ \frac{ \sigma_{\mathbf{y} \left( i \right)}^2}{\sigma_{\mathbf{x} \left( i \right)}^2} \right]
	-  \sum_{i=1}^m \ln  \left[ \frac{ \sigma_{\mathbf{y} \left( i \right)}^2 }{\sigma_{\mathbf{x} \left( i \right)}^2} \right] - m \right\},
	\nonumber
	\end{flalign}
	where equality holds if $\Sigma_{\mathbf{y}}$ is a diagonal matrix as
	\begin{flalign}
	\Sigma_{\mathbf{y}}
	= \Lambda_{\mathbf{y}}
	= \mathrm{diag} \left( \sigma_{\mathbf{y} \left( 1 \right)}^2, \ldots, \sigma_{\mathbf{y} \left( m \right)}^2 \right).
	\end{flalign}
\end{proposition}

\vspace{3mm}

\begin{proof}
	Note first that $\Sigma_{\mathbf{x}}$ is a diagonal matrix means that $\mathbf{x} \left( 1 \right), \ldots, \mathbf{x} \left( m \right)$ are mutually independent.
	It then follows from Proposition~\ref{p2} that 
	\begin{flalign}
	\mathrm{KL} \left( p_{\mathbf{y}} \| p_{\mathbf{x}} \right)
	\geq \sum_{i=1}^{m} \mathrm{KL} \left( p_{\mathbf{y} \left( i \right)} \| p_{\mathbf{x} \left( i \right)} \right),
	\nonumber
	\end{flalign}
	where equality holds if $\mathbf{y} \left( 1 \right), \ldots, \mathbf{y} \left( m \right)$ are mutually independent, i.e., if $\Sigma_{\mathbf{y}}$ is a diagonal matrix. In addition, since $\mathbf{x}$ and $\mathbf{y}$ are Gaussian (and hence $\mathbf{x} \left( 1 \right), \ldots, \mathbf{x} \left( m \right)$ and $\mathbf{y} \left( 1 \right), \ldots, \mathbf{y} \left( m \right)$ are all Gaussian), it is known from Lemma~\ref{Gaussian} that
	\begin{flalign}
	\mathrm{KL} \left( p_{\mathbf{y}} \| p_{\mathbf{x}} \right) = \frac{1}{2} \left[ \mathrm{tr} \left( \Sigma_{\mathbf{y}} \Sigma_{\mathbf{x}}^{-1} \right) - \ln \det \left( \Sigma_{\mathbf{y}} \Sigma_{\mathbf{x}}^{-1} \right) - m \right],
	\nonumber
	\end{flalign}
	and
	\begin{flalign}
	\mathrm{KL} \left( p_{\mathbf{y} \left( i \right)} \| p_{\mathbf{x} \left( i \right)} \right) = \frac{1}{2} \left\{  \frac{ \sigma_{\mathbf{y} \left( i \right)}^2}{\sigma_{\mathbf{x} \left( i \right)}^2} 
	-  \ln  \left[ \frac{ \sigma_{\mathbf{y} \left( i \right)}^2 }{\sigma_{\mathbf{x} \left( i \right)}^2} \right] - 1 \right\}.
	\nonumber
	\end{flalign}
	This concludes the proof.
\end{proof}

Proposition~\ref{p3} is essentially a property of positive definite matrices that goes beyond the KL divergence. To state it without relating to KL divergence, for any 
two positive definite  matrices $\Sigma_{\mathbf{x}}$ and $\Sigma_{\mathbf{y}}$, if  $\Sigma_{\mathbf{x}}$ is a diagonal matrix, then it holds that
\begin{flalign}
&\frac{1}{2} \left[ \mathrm{tr} \left( \Sigma_{\mathbf{y}} \Sigma_{\mathbf{x}}^{-1} \right) - \ln \det \left( \Sigma_{\mathbf{y}} \Sigma_{\mathbf{x}}^{-1} \right) - m \right] \nonumber \\
&~~~~ \geq \frac{1}{2} \left\{ \sum_{i=1}^{m} \left[ \frac{ \sigma_{\mathbf{y} \left( i \right)}^2}{\sigma_{\mathbf{x} \left( i \right)}^2} \right]
-  \sum_{i=1}^m \ln  \left[ \frac{ \sigma_{\mathbf{y} \left( i \right)}^2 }{\sigma_{\mathbf{x} \left( i \right)}^2} \right] - m \right\},
\nonumber
\end{flalign}
or equivalently,
\begin{flalign}
&\mathrm{tr} \left( \Sigma_{\mathbf{y}} \Sigma_{\mathbf{x}}^{-1} \right) - \ln \det \left( \Sigma_{\mathbf{y}} \Sigma_{\mathbf{x}}^{-1} \right) \nonumber \\
&~~~~ \geq \sum_{i=1}^{m} \left[ \frac{ \sigma_{\mathbf{y} \left( i \right)}^2}{\sigma_{\mathbf{x} \left( i \right)}^2} \right]
-  \sum_{i=1}^m \ln  \left[ \frac{ \sigma_{\mathbf{y} \left( i \right)}^2 }{\sigma_{\mathbf{x} \left( i \right)}^2} \right],
\nonumber
\end{flalign}
where the equalities hold if $\Sigma_{\mathbf{y}}$ is a diagonal matrix.

%
%
%

%

Finally, note that Proposition~\ref{c1} can be proved by invoking Proposition~\ref{p1}, Proposition~\ref{p2}, and Proposition~\ref{p3} in order. Overall, it is concluded that that independent Gaussian distributions minimize the KL divergence from given independent Gaussian distributions.

\section{Conclusion}

We have presented a property of the KL divergence in this short note: Independent Gaussian distributions minimize the KL divergence from given independent Gaussian distributions. It might be interesting to examine the implications of this property, as in \cite{KLarxiv}.

\balance

\bibliographystyle{IEEEtran}
\bibliography{references}

\end{document}